\documentclass[letterpaper]{article}
\usepackage{authblk}
\usepackage{times}
\usepackage{helvet}
\usepackage{courier}
\usepackage{graphicx}
\usepackage{subfig}
\usepackage{amsfonts}
\usepackage{amsmath}
\usepackage{multirow}
\usepackage{amssymb}
\usepackage{amsthm}
\newtheorem{theorem}{Theorem}

\newtheorem{lemma}{Lemma}

\newtheorem{definition}{Definition}

\begin{document}

%\bibliographystyle{plain}

%\title{Probabilistic Analysis of Online Stacking Algorithms}
\title{An Asymptotically Optimal Algorithm for Online Stacking\footnote{A preliminary version of this paper appeared in the proceedings of the 6th International Conference on Computational Logistics, ICCL '15, under the title "Probabilistic Analysis of Online Stacking Algorithms"~\cite{DBLP:conf/iccl2/OlsenG15}.}}

\author{Martin Olsen}
\author{Allan Gross}
\affil{Department of Business Development and Technology \authorcr Aarhus University \authorcr Denmark \authorcr \{\tt martino, agr\}@btech.au.dk}
%\author{Martin Olsen and Allan Gross\\Department of Business Development and Technology\\Aarhus University\\Denmark\\martino@btech.au.dk\\ }

\maketitle

\begin{abstract}
  Consider a storage area where arriving items are stored temporarily in bounded capacity stacks
  until their departure. We look into the problem of deciding where to
  put an arriving item with the objective of minimizing the maximum number of stacks used over time. The decision has to be made as soon as an item arrives, and
  we assume that we only have information on the departure times for
  the arriving item and the items currently at the storage area. We
  are only allowed to put an item on top of another item if the item
  below departs at a later time. We refer to this problem as online stacking.
  We assume that the storage time intervals are picked i.i.d. from $[0, 1] \times [0, 1]$ using an unknown distribution with a bounded probability density function. Under this mild condition, we present a simple polynomial
  time online algorithm and show that the competitive ratio converges to $1$ in probability. The result holds if the stack
  capacity is $o(\sqrt{n})$, where $n$ is the number of items, including the realistic case where the capacity is a constant. Our experiments show that our results also have practical relevance. To the best of our knowledge, we are the first to present an asymptotically optimal algorithm for online stacking, which is an important problem with many real-world applications within computational logistics.

%\keywords{Online algorithms \and Stacking \and Stowage \and Asymptotic optimality}
\end{abstract}
%\begin{center} {\bf Keywords}: 
%\end{center}

\section{Introduction}

In this paper, we consider the situation that some items arrive at a
storage location where they are temporarily stored in LIFO stacks
until their departure. When an item arrives, we are faced with the
problem of deciding where to store the item. We will refer to this
problem as the stacking problem. The stacking problem has many
applications within real-world logistics. As an example, the items
could be containers, and the storage location could be a container
terminal or a container ship~\cite{Borgman2010}. The items could also
be steel bars~\cite{Rei13} and trains~\cite{Demange2012}, or the storage
location could simply be a warehouse storing anything that could be
stacked on top of each other.

We focus on the variant of the stacking problem given by the following
assumptions: 1) We have to make a decision on where to store an item
as soon as it arrives. When an item $i$ arrives at time $x_i$, we are
informed on the departure time $y_i$ of the item, but we have no
information on future items. In other words, we look at an {\em
  online} version of the problem, and we look for online algorithms
solving the problem. 2) The numbers $x_i$ and $y_i$ could be any real
numbers. This means that we restrict our attention to what we will
refer to as the {\em continuous} case as opposed to the {\em discrete}
case, where we only have a few possibilities for $x_i$ and $y_i$. 3) We
are only allowed to put an item $i$ on top of an item $j$ if $y_i \leq
y_j$. Another way of saying this is that we do not allow {\em
  rehandling}/{\em relocations}/{\em overstowage} of items. 4) The objective is to
minimize the maximum number of stacks in use over time given a bound $h$ on the stacking
height.
%There are several real world applications for algorithms
%solving this variant for the stacking problem. Two examples for
%applications are track assignment for trains at train
%depots~\cite{Demange2012} and stacking of containers in container
%terminals.

\subsection{Contribution}\label{sec:contribution}

We use the {\em unknown distribution model} for generating
stacking problem instances, where the time intervals for storing the items are picked i.i.d. using an unknown distribution with bounded density:
\begin{definition}\label{def:unknown}
{\bf The Unknown Distribution Model}: Let $n$ pairs $(a, b) \in [0, 1] \times [0, 1]$ be drawn i.i.d. using an unknown distribution with a bounded probability density function. For each pair $(a, b)$, let an item arrive at the storage area at time $x = \min(a, b)$ and leave the storage area at time $y = \max(a,b)$.
\end{definition}
If the reader prefers a model satisfying $a < b$, we can use a density $f$ with $f(a, b)=0$ for $a \geq b$. It is very common to use distributions with bounded densities to model real scenarios. For the univariate case, some examples of such distributions are normal distributions (also called Gaussian distributions), uniform distributions, triangular distributions, and exponential distributions. Assuming independence seems to be reasonable when items arrive at the storage area from different sources. This shows that our model is applicable for many realistic scenarios.

The main contribution of our paper is a simple polynomial time online algorithm called, for the lack of a better name, Algorithm B
for which the following holds for stack capacity $h = o(\sqrt{n})$ including the realistic case where $h$ is a fixed constant: For any positive real numbers $\epsilon_1, \epsilon_2 > 0$,
there exists an $N > 0$ such that Algorithm B uses no more than
$(1+\epsilon_1)\chi_h$ stacks with probability at least $1-\epsilon_2$, if
the number of items is at least $N$, where $\chi_h$ denotes the optimal
number of stacks. In other words, we show that the competitive ratio
of Algorithm B {\em converges to $1$ in probability} if $h = o(\sqrt{n})$:
$$\frac{\chi^\prime_h}{\chi_h} \xrightarrow{p} 1 \mbox{ for } n \rightarrow \infty$$
where $\chi^\prime_h$ denotes the number of stacks used by the algorithm. If $h$ is a constant, then the expected value of the competitive ratio for Algorithm B converges to $1$ in the standard sense of convergence:
$$E\left(\frac{\chi^\prime_h}{\chi_h}\right) \rightarrow 1 \mbox{ for } n \rightarrow \infty \enspace .$$

These results are corollaries of the main theorem of our paper:
\begin{theorem}
\label{thm:maintheorem}
For the unknown distribution model, Algorithm B produces a solution for the online stacking problem ($h$-OVERLAP-COLORING) such that
\begin{equation}
\frac{\chi^\prime_h}{\chi_h} \leq 1 + O(hn^{-\frac{1}{2}}) \mbox{ whp .}
\end{equation}
Algorithm B processes one item in $O(\log n)$ time.
\end{theorem}

To the best of our knowledge, we are the first to present an asymptotically optimal polynomial time online algorithm for stacking -- an offline version has not been presented either. Similar algorithms like Algorithm B have been presented earlier in the literature~\cite{Borgman2010,Duinkerken2001,HAMDI2012357,Wang2014}, so the most important part of the contribution is the formal proof of asymptotic optimality under mild conditions.

We also verify the results experimentally using two types of distributions and instances with $2000 \leq n \leq 200000$ and $h=5$. For {\underline {all}} our instances, $\frac{\chi^\prime_h}{\chi_h} \leq 1 + kn^{-\frac{1}{2}}$ for a moderate constant $k$ depending on the distribution involved, indicating that our results also have practical importance. 

\subsection{Related Work}\label{sec:relatedwork}

A preliminary version of this paper~\cite{DBLP:conf/iccl2/OlsenG15} was presented at the conference ICCL 2015. The results in the present version are more generic and stronger since they are based on the unknown distribution model as compared to the results obtained in the preliminary version, which were based on a uniform distribution on the input. The present version furthermore includes a section with experiments.

%As mentioned earlier, Algorithm B has been presented in a preliminary version of this paper~\cite{DBLP:conf/iccl2/OlsenG15} where the algorithm was shown to be asymptotically optimal for the restrictive {\em uniform model} where the times for arrival and departure for the items are picked independently and uniformly from the interval $[0, 1]$. In this extended version that also includes experiments we show asymptotic optimality assuming the much milder and realistic unknown distribution model.

The offline variant of the stacking problem where all information is
provided before any decisions are made is NP-hard for any fixed bound
$h \geq 6$ on the stacking height~\cite{Cornelsen2007} as can be seen
by reduction from the coloring problem on permutation
graphs~\cite{Jansen2003}. To the best of our knowledge, the computational complexity for the case $2 \leq h \leq 5$ is an open problem for the offline case. This variant of the problem is also NP-hard
in the unbounded case as shown by Avriel et
al.~\cite{Avriel2000}. Tierney et al.~\cite{kshift-dam} show that the
problem of deciding if it is possible to accommodate all the items in a
{\em fixed} number of bounded capacity stacks without relocations can
be solved in polynomial time, but the running time of their offline
algorithm is huge even for a small (fixed) number of stacks.

Cornelsen and Di Stefano~\cite{Cornelsen2007} and Demange et
al.~\cite{Demange2012} consider the problem in the context of
assigning tracks to trains arriving at a train
station/depot. Cornelsen and Di Stefano look at unbounded capacity
stacks (train tracks) whereas Demange et al. consider unbounded as well as
bounded capacity stacks. For unbounded stack capacity, Demange et al. show that no online stacking algorithm has a constant competitive ratio. In addition, they present lower and upper bounds for
the competitive ratio with some improvements added later by Demange and Olsen~\cite{DBLP:conf/walcom/DemangeO18}. For bounded capacity stacks, Demange et al.~\cite{Demange2012} present lower and upper bounds
around $2$ for the competitive ratio for online stacking restricted to the
situation where all trains are at the train depot at some point in
time. This condition is known as {\em the midnight condition}. It is well-known that the stacking problem can be be solved exactly and online in polynomial time for the unbounded stack capacity case with the midnight condition by using the Patience Sorting method presented later in this paper.

Simple heuristics for online stacking similar to Algorithm B have been presented by Borgman et al.~\cite{Borgman2010}, Duinkerken et al.~\cite{Duinkerken2001}, Hamdi et al.~\cite{HAMDI2012357}, and Wang et al.~\cite{Wang2014} without providing a proof of asymptotic optimality. 
Finally, we mention the work of Rei and Pedroso~\cite{Rei13} and
K\"onig et al.~\cite{Konig2007} on related problems within the steel
industry as well as the PhD thesis by Pacino~\cite{PacinoPhD} on container
ship stowage.

\subsection{Outline of the Paper}

In Section~\ref{sec:preliminairies}, we look at the link between
stacking problems and the coloring problems for overlap graphs and
interval graphs and introduce some terminology used in this paper. We
also consider some results from the field of probability theory that
form the basis for the probabilistic analysis of our online
algorithm. Our algorithm is introduced in an offline and an online
version in Section~\ref{sec:algorithm}. The
analysis of the algorithm and our main result are presented in
Section~\ref{sec:probabilistic_analysis}, and finally, we verify our results experimentally in Section~\ref{sec:experiments}.

\section{Preliminaries}\label{sec:preliminairies}

In this section, we present most of the terminology used in this paper
and some results from probability theory, which we will use later.

\subsection{Connections to Graph Coloring}\label{sec:coloring}

For each item $i$, we have an interval $I_i=[x_i, y_i]$ specifying the
time interval for which the item has to be temporarily stored. To make
it easier to formulate the constraint on the stacking height, we assume realistically
that items cannot arrive and depart at exactly the same time. This assumption is consistent with the unknown distribution model that generates 
storage time intervals having pairwise distinct endpoints with probability $1$.

It is well-known that the problem we consider can be formulated
as a graph coloring problem~\cite{Avriel2000}, and we will use graph coloring terminology
in the remaining part of the paper in order to make the presentation
generic. We say that two intervals $I_1=[x_1, y_1]$ and $I_2=[x_2,
y_2]$ {\em overlap} if and only if $x_1 < x_2 < y_1 < y_2$ or $x_2 <
x_1 < y_2 < y_1$. We can put an item on top of another item if and
only if their corresponding intervals do not overlap so our problem
can now be formally defined as follows, where $h$ is the maximum
allowed stack height:
\begin{definition}\label{def:problem} The $h$-OVERLAP-COLORING
  problem:
\begin{itemize}
\item Instance: A set of $n$ intervals $\mathcal{I} = \{I_1, I_2,
  \ldots , I_n\}$, where all the endpoints of the intervals are
  distinct.
\item Solution: A coloring of the intervals using a minimum number of
  colors such that the following two conditions hold:
\begin{enumerate}
\item Any two overlapping intervals have different colors.
\item For any real number $t$ and any color $d$, there will be no more
  than $h$ intervals with color $d$ that contain $t$.
\end{enumerate}
\end{itemize}
\end{definition}
It should be stressed that we look for online algorithms that process the intervals in order of increasing starting
points.

The problem can be viewed as a graph coloring problem for the graph
with a vertex for each interval and an edge between any two vertices
where the corresponding intervals overlap. Such a graph is known as an
{\em overlap graph}. As mentioned earlier, we let $\chi_h$ denote the minimum number of
colors for a solution.

An {\em interval graph} is a graph in which each vertex corresponds to an
interval and with an edge between two vertices if and only if the
corresponding intervals {\em intersect}. It is well-known that we can
obtain a minimum coloring of an interval graph if we use the following
simple online algorithm to process the intervals in increasing order
of their starting points: If we can reuse a color, we do so --
otherwise we pick a new color that we have not used previously. The
clique number of a graph is the size of a maximum clique. Interval
graphs are members of the family of {\em perfect} graphs, implying that
all interval graphs can be colored with a number of colors
corresponding to their clique number.

\subsection{Increasing Subsequences and Patience Sorting}\label{sec:patiencesorting}

The algorithm we present in Section~\ref{sec:algorithm} and the
probabilistic analysis performed in
Section~\ref{sec:probabilistic_analysis} are based on some results from
the theory on increasing subsequences and the method of Patience
Sorting, which we will introduce next. Patience Sorting~\cite{Aldous99}
is a method originally invented for sorting a deck of cards. Now
imagine that we have a small deck of cards as follows, where the top of
the deck is the leftmost card (the underlined cards will be explained
later):
$$9, \underline{2}, \underline{4}, 8, 1, 7, \underline{6}, 3, 5, \underline{10}$$
We take the top card $9$ and start a new pile. We now remove the
other cards from the initial deck one by one from the top of the
deck. Each time we remove a card, we try to put it in another pile with
a top card of higher value than the removed card. If possible, we choose a pile where the top card has the lowest
value. If not, we start a new pile. Card $2$ goes on top of
card $9$ but we have to start two new piles with cards $4$ and $8$,
respectively. Card $1$ can be put on top of card $2$,
etc. Finally, we face the following four piles:
$$1, 2, 9 \enspace \enspace \enspace \enspace 3, 4 \enspace \enspace \enspace \enspace 5, 6, 7, 8 \enspace \enspace \enspace \enspace 10$$
It is now easy to sort the cards by repeatedly picking the top
card with the lowest value. This is the Patience Sorting method, and we refer the reader to
the work by Aldous~\cite{Aldous99} for more details.

Let $L_n$ be the random variable representing the resulting number of
piles for the Patience Sorting method on a deck with $n$ cards. It is
worth noting that $L_n$ is identical to the length of the longest
increasing subsequence for the sequence of cards defined by the
deck. To illustrate this, there are several increasing subsequences
that have length $4$ for the sequence shown above (for example, the
subsequence $2$, $4$, $6$, $10$, which is underlined) but no increasing
subsequence with length $5$ or more -- and the number of piles needed
is $4$. Each pile represents a decreasing subsequence, and $L_n$ is the
minimum number of decreasing subsequences into which the sequence can
be partitioned. Let $\mu$ and $\sigma$ denote the expected value and
the standard deviation of $L_n$ respectively, under the assumption that
the permutation corresponding to the deck of cards is picked uniformly
at random. The asymptotic behavior of $L_n$ is described as follows,
where $\sigma_{\infty}$ is a positive
constant~\cite{Aldous99,Pilpel1990}:
\begin{equation}
\label{eq:mu}
\mu \leq 2\sqrt{n}
\end{equation}

\begin{equation}
\label{eq:sigma}
\sigma = \sigma_{\infty}n^{\frac{1}{6}}+o(n^{\frac{1}{6}})
\end{equation}
These facts are crucial for the analysis of the online algorithm we
present later in this paper.

\section{The Algorithm}\label{sec:algorithm}

Before we present our stacking strategy, we need to introduce a little
more terminology. A {\em chain} of intervals is a sequence of
intervals $I_1 \supseteq I_2 \supseteq I_3 \supseteq \ldots \supseteq
I_m$. The intervals in a chain represent items that may be stacked on
top of each other. We refer to the intervals $I_1$ and $I_m$ as the
{\em bottom} and the {\em top} of the chain, respectively. For a given
number $h$, we can split a chain into chains of cardinality $h$ or
less in a natural way: The intervals $I_1$ to $I_h$ form the first
chain, the next $h$ intervals $I_{h+1}$ to $I_{2h}$ form the next
chain, etc. A partition of $\mathcal{I}$ into chains is a set of
chains such that each interval is a member of exactly one chain.

We present two versions of our algorithm (named A and B), which produce
the same coloring for any instance of the $h$-OVERLAP-COLORING
problem. Algorithm A is an offline version, and Algorithm B is an
online version. Algorithm A is presented in order to make it easier
for the reader to understand the coloring strategy used.

We are now ready to describe Algorithm A, which consists of $4$ steps
listed in Fig.~\ref{fig:algA}. In the first step, we partition
$\mathcal{I}$ into a minimum number of chains as illustrated in
Fig.~\ref{fig:alg_initial}.  In the second step, we split the chains
into chains of cardinality $h$ or less as described above. The
interval graph of the bottoms of the chains is colored in the third
step using the simple algorithm described in
Section~\ref{sec:coloring}. Finally, in the fourth step, all the
remaining intervals are colored with the color at the bottom of their
chain. Steps $2$, $3$, and $4$ are illustrated in
Fig.~\ref{fig:alg_finish} for the case $h=2$.  It is not hard to see
that the coloring produced satisfies the conditions from
Definition~\ref{def:problem}: All the chains produced in step $2$ have
cardinality at most $h$, and chain bottoms with the same color do not
intersect.

\begin{figure}
  \begin{center}
    \mbox{\begin{minipage}[t]{0cm}
      \begin{tabbing}
        xxx\=xxxxxxx\=xxx\=xxx\=xxx\=xxx\=\kill
        \textbf{Algorithm A($\mathcal{I}$, $h$)}: \\
        \> Step 1:\>  Partition $\mathcal{I}$ into a minimum number $c$ of chains. \\ 
        \> Step 2:\>  Split the chains into chains of cardinality $h$ or less. \\
        \> Step 3:\>  Color the interval graph formed by the bottoms of the chains with \\
        \>\> $\chi^\prime_h$ colors. \\
        \> Step 4:\>  Color any interval not at the bottom of a chain with the color of \\
        \>\>  the bottom of its chain. \\
     \end{tabbing}
   \end{minipage}}
  \end{center}
  \caption{The offline version of our algorithm.}
  \label{fig:algA}
\end{figure}

\begin{figure}
  \centering \subfloat{\label{fig:alg_instance}\includegraphics[scale=0.5]{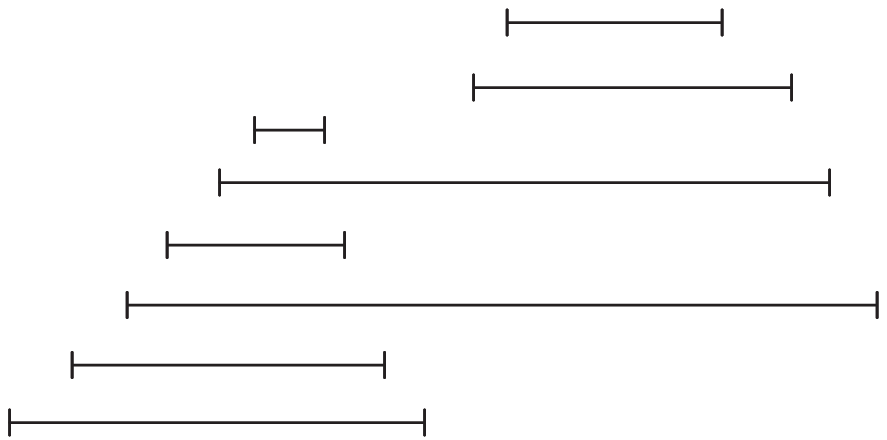}}
  \hspace{1cm} \subfloat{\label{fig:alg_step1}\includegraphics[scale=0.5]{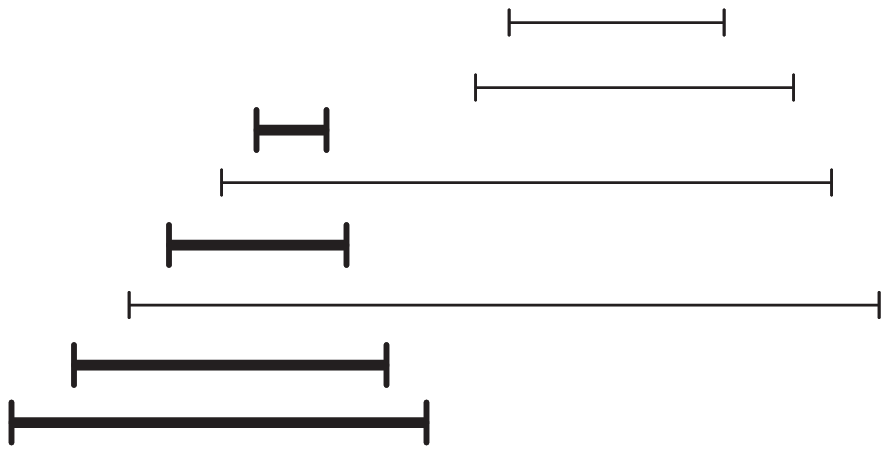}}
  \caption{The initial phase of Algorithm A is illustrated here. To
    the left, we see the intervals forming the instance. The two chains
    created in step $1$ are shown to the right.}
\label{fig:alg_initial}
\end{figure}

\begin{figure}
  \centering
  \subfloat{\label{fig:alg_step2}\includegraphics[scale=0.5]{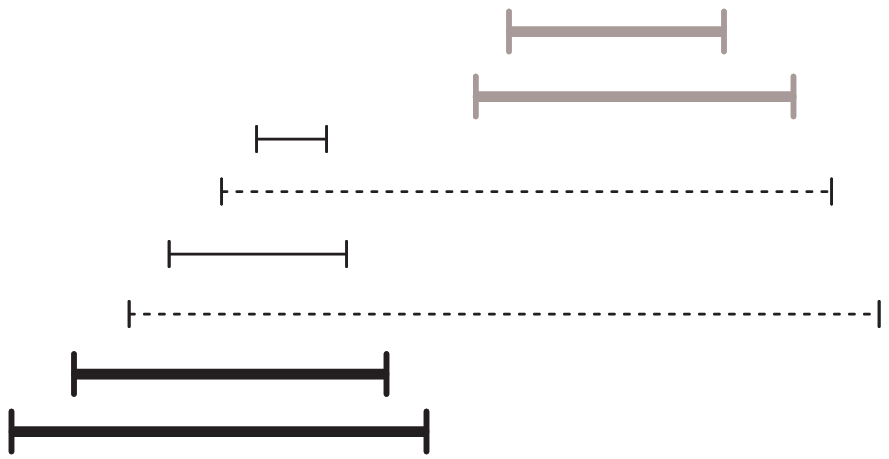}}
  \hspace{1cm}
  \subfloat{\label{fig:alg_step3_4}\includegraphics[scale=0.5]{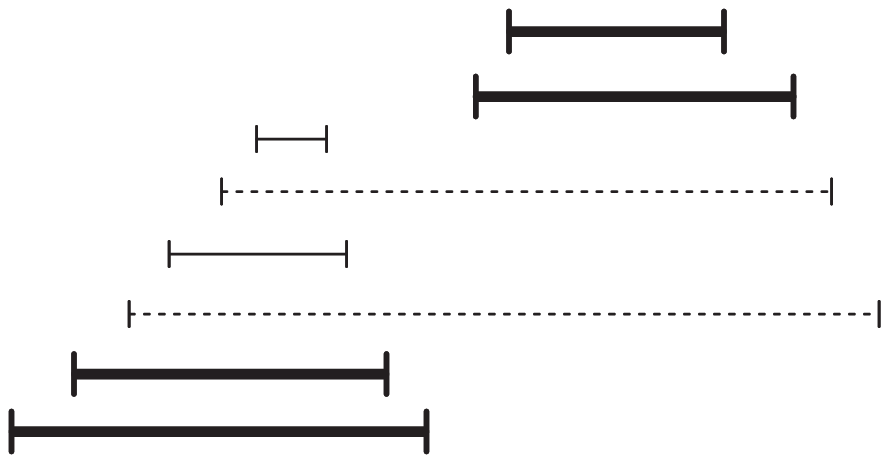}}
  \caption{This figure illustrates the final phase of Algorithm A for
    the case $h=2$. The four chains produced in step $2$ are shown to
    the left, and the coloring produced in steps $3$ and $4$ is
    shown to the right. The algorithm generates a coloring using
    $\chi^\prime_h=3$ colors.}
\label{fig:alg_finish}
\end{figure}

We now prove that it is possible to transform Algorithm A into an
online version, Algorithm B, which is listed in Fig.~\ref{fig:algB}.
\begin{lemma}
\label{lem:online_implementation}
Algorithm B is an online algorithm for the h-OVERLAP-COLORING problem
producing a coloring identical to the coloring produced by
Algorithm A. Algorithm B processes one interval in $O(\log n)$ time.
\end{lemma}

\begin{proof}
  Let $\pi$ be a permutation of the integers from $1$ to $n$ such that
  $x_{\pi(i)} < x_{\pi(j)}$ for $i < j$. Now we consider the sequence
  where the $i$'th number is $y_{\pi(i)}$. There is a simple
  one-to-one correspondence between a decreasing subsequence of this
  sequence and a chain of intervals from the set $\mathcal{I}$: If we
  start at the bottom of a chain and move upward, then the $x$-values
  increase and the $y$-values decrease. This means that we obtain a
  partition of $\mathcal{I}$ into a minimum number of chains, if we
  apply the Patience Sorting method described in
  Section~\ref{sec:patiencesorting} and partition the sequence into a
  minimum number of decreasing subsequences.

  Algorithm B processes the intervals in increasing order of their starting
  points applying the Patience Sorting method, and decisions on an
  interval are made without considering intervals with bigger starting
  points. The same goes for the splitting into smaller chains as well as the
  coloring of the chain bottoms and the other intervals. This means
  that Algorithm B is an online algorithm producing the same
  coloring as Algorithm A.

  Each step of the Patience Sorting method requires $O(\log n)$ time
  if we use binary search to locate the right pile. Keeping track of
  unused colors can also be handled in $O(\log n)$ time for each step
  if a priority queue is used (a priority queue storing information on when
  the colors expire is used for the set
  $\mathcal{D}$ in Fig.~\ref{fig:algB}).
%\hfill \qed
\end{proof}

\begin{figure}
  \begin{center}
    \mbox{\begin{minipage}[t]{0cm}
      \begin{tabbing}
        xxx\=xxx\=xxx\=xxx\=xxx\=xxx\=\kill
        \textbf{Algorithm B($\mathcal{I}$, $h$)}: \\
        \> Assumption on $\mathcal{I} = \{ [x_1, y_1], [x_2, y_2], \ldots , [x_n, y_n]\}$: $i < j \Rightarrow x_i < x_j$ \\
        \> \ 1:\>  $\mathcal{C} \leftarrow \emptyset$ \\
        \> \ 2:\>  $\mathcal{D} \leftarrow \emptyset$ \\
        \> \ 3:\>  $\chi \leftarrow 0$ \\
        \> \ 4:\>  {\bf for} $i \in \{1, 2, \ldots, n\}$ {\bf do}\\
        \> \ 5:\>\>  $bottom \leftarrow$ false \\
        \> \ 6:\>\>  Let $\mathcal{H}$ be the set of chains in $\mathcal{C}$ where \\
        \>\>\>  the top of the chain contains $I_i$. \\
        \> \ 7:\>\>  {\bf if} $\mathcal{H} = \emptyset$ {\bf then}\\
        \> \ 8:\>\>\>  Add a new chain to $\mathcal{C}$ consisting of $I_i$. \\
        \> \ 9:\>\>\>  $bottom \leftarrow$ true \\
        \> 10:\>\>  {\bf else} \\
        \> 11:\>\>\>  Let $c_J$ be the chain in $\mathcal{H}$ with a top interval \\
        \>\>\>\>  $J[x_J, y_J]$ with the smallest value of $y_J$. \\
        \> 12:\>\>\>  Put $I_i$ on top of $c_J$. \\
        \> 13:\>\>\>  Let $d$ be the color assigned to $J$. \\
        \> 14:\>\>\>  {\bf if} there are less than $h$ intervals in $c_J$ with color $d$ {\bf then} \\
        \> 15:\>\>\>\>  Assign color $d$ to $I_i$. \\
        \> 16:\>\>\>  {\bf else} \\
        \> 17:\>\>\>\>  $bottom \leftarrow$ true \\
        \> 18:\>\>\>  {\bf end if} \\
        \> 19:\>\>  {\bf end if} \\
        \> 20:\>\>  \\
        \> 21:\>\>  {\bf if} $bottom = $ true {\bf then}\\
        \> 22:\>\>\>  Let $\mathcal{G} = \{ (d, y) \in \mathcal{D}: y < x_i \}$ \\
        \> 23:\>\>\>  {\bf if} $\mathcal{G} = \emptyset$ {\bf then}\\
        \> 24:\>\>\>\>  $\chi \leftarrow \chi + 1$ \\
        \> 25:\>\>\>\>  Assign color $\chi$ to $I_i$.\\
        \> 26:\>\>\>\>  $\mathcal{D} \leftarrow \mathcal{D} \cup \{ (\chi, y_i) \} $ \\
        \> 27:\>\>\>  {\bf else} \\
        \> 28:\>\>\>\>  Pick any $(d, y) \in \mathcal{G}$. \\
        \> 29:\>\>\>\>  Assign color $d$ to $I_i$. \\
        \> 30:\>\>\>\>  $\mathcal{D} \leftarrow (\mathcal{D} \setminus \{ (d, y) \}) \cup \{ (d, y_i) \}$ \\
        \> 31:\>\>\>  {\bf end if} \\
        \> 32:\>\>  {\bf end if} \\
        \> 33:\>  {\bf end for} \\
     \end{tabbing}
   \end{minipage}}
  \end{center}
  \caption{The online version of our algorithm. Please note that we
    assume the intervals in $\mathcal{I}$ to appear in increasing
    order of their starting points.}
  \label{fig:algB}
\end{figure}

\section{Probabilistic Analysis}\label{sec:probabilistic_analysis}

Let $\omega^\prime$ be the clique number of the interval graph formed
by the set of intervals $\mathcal{I}$. We remind the reader that $c$
is the minimum number of chains formed in step $1$ of Algorithm A.

\begin{lemma} 
\label{lem:chi_prime_h}
The coloring produced by Algorithm A and B uses $\chi^\prime_h$
colors satisfying
\begin{equation}
  \chi^\prime_h \leq \frac{\omega^\prime}{h} + c \enspace .
\end{equation}
\end{lemma}

\begin{proof}
  For any real number $x$, we let $g_x$ denote the number of intervals
  in $\mathcal{I}$ that contain $x$ and $g^h_x$ denote the number of
  chain bottoms produced in step $2$ of Algorithm A containing $x$. As mentioned in
  Section~\ref{sec:coloring}, any interval graph can be colored with a
  number of colors corresponding to the size of the largest clique of
  the graph:
\begin{equation}
\label{eq:chi_prime_h_1}
\chi^\prime_h = \max_{x} g^h_x \enspace .
\end{equation}
Now consider an interval that is a bottom of a chain produced in step
$2$ of Algorithm A but not a bottom of one of the chains produced in
step $1$. If such an interval contains a number $x$, then the $h-1$
intervals directly below it in the chain will also contain $x$. There
are at least $g^h_x - c$ such intervals that contain $x$ so we obtain
the following inequality:
\begin{equation}
  (g^h_x - c)h \leq g_x \enspace .
\end{equation}
We now rearrange this inequality:
\begin{equation}
\label{eq:chi_prime_h_2}
\max_{x} g^h_x \leq \frac{\max_{x} g_x}{h} + c \enspace .
\end{equation}
Next, we use~(\ref{eq:chi_prime_h_1}) and $\omega^\prime = \max_{x}
g_x$.
%\hfill \qed
\end{proof}
Our aim is
to show that the competitive ratio $\frac{\chi^\prime_h}{\chi_h}$ of
Algorithm B is close to $1$ with high probability. Formally, we say
that an event $E_n$ occurs {\em with high probability}, abbreviated
whp, if $P(E_n)~\rightarrow~1$ for $n~\rightarrow~\infty$. There is a number $t$
contained in $\omega^\prime$ intervals implying $\chi_h \geq
\frac{\omega^\prime}{h}$. Using Lemma~\ref{lem:chi_prime_h}, we can
conclude that the competitive ratio is not bigger than
$1+c/\left(\frac{\omega^\prime}{h}\right)$. We will now show that the
competitive ratio is $1 + O(hn^{-\frac{1}{2}})$ whp. The strategy of
our proof is to show that $c = O(\sqrt{n})$ whp and that $\omega^\prime = \Omega(n)$ whp, and then combine these results.

For a brief moment, we leave the unknown distribution model and present a lemma for a simpler model for generating the instances: {\em the uniform model}. This model is obtained by substituting the unknown distribution in the unknown distribution model (see Definition~\ref{def:unknown}) with the uniform distribution on $[0, 1] \times [0, 1]$. This is the only place in the paper where we are {\em not} using the unknown distribution model.

\begin{lemma}
\label{lem:uniform}
For the uniform model, the set of intervals $\mathcal{I}$ can be partitioned into $c$ chains such that
$$c \leq 5\sqrt{n} \mbox{ whp} \enspace .$$
\end{lemma}

\begin{proof}
  Let $(a_i, b_i)$ denote the $i$'th pair drawn using the uniform model. We introduce a permutation
  $\pi^\prime$ on the integers from $1$ to $n$ defined by
  $a_{\pi^\prime(i)} < a_{\pi^\prime(j)}$ for $i < j$. We now look at
  the sequence of $b$-values with $b_{\pi^\prime(i)}$ as the $i$'th
  number in the sequence. We use the Patience Sorting method from
  Section~\ref{sec:patiencesorting} on the $b$-sequence and obtain
  $c^\prime$ decreasing subsequences. We split each subsequence into
  two decreasing subsequences if there is a point where the $b$-values
  become lower than their corresponding $a$-values. It is not hard
  to see that we can form a chain of intervals for each of the up to
  $2c^\prime$ subsequences we obtain by the splitting procedure (see
  Fig.~\ref{fig:uniform}).
\begin{figure}
  \centering 
  \includegraphics[scale=0.5]{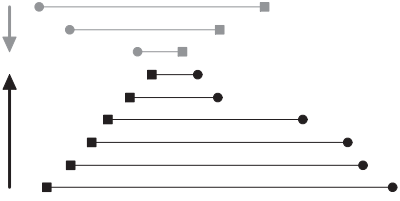}
  \caption{The figure shows a decreasing subsequence for the sequence
    of $b$-values. The squares and circles correspond to $a$-values
    and $b$-values, respectively. The decreasing subsequence can be
    split into a grey chain and a black chain of intervals.}
  \label{fig:uniform}
\end{figure}
Since $c \leq 2c^\prime$, we have the following:
\begin{equation}
\label{eq:UpperBound1}
P(c \geq 5\sqrt{n}) \leq P\left(c^\prime \geq \frac{5}{2} \sqrt{n}\right) \enspace . 
\end{equation}
The $a$- and $b$-values are independent for the uniform model, so $c^\prime$ and $L_n$ have the same distribution, where
$L_n$ is the length of the longest increasing subsequence for a
permutation of $n$ numbers chosen uniformly at random (see
Section~\ref{sec:patiencesorting}):
\begin{equation}
\label{eq:UpperBound1.5}
P\left(c^\prime \geq \frac{5}{2} \sqrt{n}\right) = P\left(L_n \geq \frac{5}{2} \sqrt{n}\right) \enspace .
\end{equation}
Using~(\ref{eq:mu}), we get the following:
\begin{equation}
\label{eq:UpperBound2}
P\left(L_n \geq \frac{5}{2} \sqrt{n}\right) \leq P\left(|L_n - \mu| \geq \frac{1}{2}\sqrt{n}\right) \enspace .
\end{equation}
From~(\ref{eq:sigma}), we observe that $\sigma \leq \frac{3}{2}\sigma_
\infty n^{\frac{1}{6}}$ for $n$ sufficiently big. By using Chebyshevs
inequality~\cite{Kobayashi2012}, we now get the following for $n$
sufficiently big:
\begin{equation}
\label{eq:UpperBound3}
P\left(|L_n - \mu| \geq \frac{1}{2}\sqrt{n}\right) \leq \frac{\sigma^2}{(\frac{1}{4}n)} \leq \frac{\frac{9}{4}\sigma_ \infty^2n^{\frac{1}{3}}}{(\frac{1}{4}n)} = 9\sigma_ \infty^2n^{-\frac{2}{3}} \enspace .
\end{equation}
From~(\ref{eq:UpperBound1}), (\ref{eq:UpperBound1.5}),
(\ref{eq:UpperBound2}), and (\ref{eq:UpperBound3}), we now get the
following for $n$ sufficiently big:
\begin{equation}
\label{eq:UpperBound4}
P(c < 5\sqrt{n}) \geq 1-9\sigma_ \infty^2n^{-\frac{2}{3}} \enspace .
\end{equation}
From (\ref{eq:UpperBound4}), we conclude that $c < 5\sqrt{n}$ whp.
%\hfill \qed
\end{proof}

We now use this lemma for the uniform model to prove a similar lemma for the more generic unknown distribution model.

\begin{lemma}\label{lem:chains}
For the unknown distribution model, the set of intervals $\mathcal{I}$ can be partitioned into $c$ chains such that
$$c = O(\sqrt{n}) \mbox{ whp} \enspace .$$
\end{lemma}
 
\begin{proof} 
Without loss of generality, we assume that the unknown distribution has a density $f$ with $f(a,b) \leq B$ with $B > 1$ (we can always increase $B$ if necessarry). Let the function $g: [0, 1] \times [0, 1] \rightarrow \mathbb{R}_+ \cup \{0\}$ be defined as follows:
$$g(a, b) = \frac{1-f(a,b)/B}{1-1/B} \enspace .$$
The function $g$ clearly qualifies as a probability density function. We now pick $n$ pairs $(a,b)$ independently by repeating the following procedure until $n$ pairs have been drawn using the $f$-distribution:
\begin{itemize}
\item Pick a pair $(a,b)$ using $f$ with probability $1/B$ or $g$ with probability $1-1/B$.
\end{itemize}
Let $w$ denote the total number of pairs picked by the procedure. Each time we pick a pair $(a,b)$, we use a mixture of the distributions $f$ and $g$: $1/B \cdot f + (1-1/B) \cdot g = 1$. This means that the $w$ pairs are picked using the uniform model described above.
Let $C$ denote the minimum number of chains that we can form for the $w$ pairs we have picked using both distributions. Using Lemma~\ref{lem:uniform}, we conclude that $C \leq 5\sqrt{w}$ whp.
If we remove an interval from a chain, the chain is still a chain. This means that it is easy to transform a set of chains for the $w$ points picked by both distributions into a set of chains for the $n$ pairs of endpoints picked using the $f$-distribution by deleting intervals picked using the $g$-distribution: $c \leq C$. 
Using the weak law of large numbers, we have $\frac{n}{w} \geq \frac{1}{2B}$ whp implying $2Bn\geq w$ whp.
Finally, we get 
\begin{equation}
\label{eq:ub}
c \leq 5\sqrt{2Bn} \mbox{ whp .}
\end{equation}
%\hfill \qed
\end{proof}
As a side remark, it should be noted that we could replace $2$ in~(\ref{eq:ub}) with any number strictly greater than $1$. This shows that the upper bound matches and extends the result for the uniform distribution ($B=1$) from Lemma~\ref{lem:uniform}.

To illustrate a case where the premises of Lemma~\ref{lem:chains} are {\em not} satisfied, we can pick a number $u$ uniformly at random from $[0, 0.9]$ and form the interval $[u, u+0.1]$. In this case, we are not using the unknown distribution model for picking the endpoints from $[0, 1] \times [0, 1]$ (there is a set with measure $0$ that has probability $1$). It is easy to see that $c = n$ with probability $1$ in this case. Ironically, our algorithm works perfectly when intervals are picked using this stochastic process.

\begin{lemma}\label{lem:lower_bound_omega} For the unknown distribution model, we have the following:
$$\omega^\prime = \Omega(n) \mbox{ whp} \enspace .$$ 
\end{lemma}
 
\begin{proof}
The triangle above the diagonal $y=x$ in the square $[0,1] \times [0,1]$ can be partitioned into squares $S^\prime_z$ as illustrated in Fig.~\ref{fig:lower_bound_omega}. The triangle below the diagonal can be partitioned in a similar way using squares $S_z$. 
\begin{figure}[h]
  \centering 
  \includegraphics[scale=0.8]{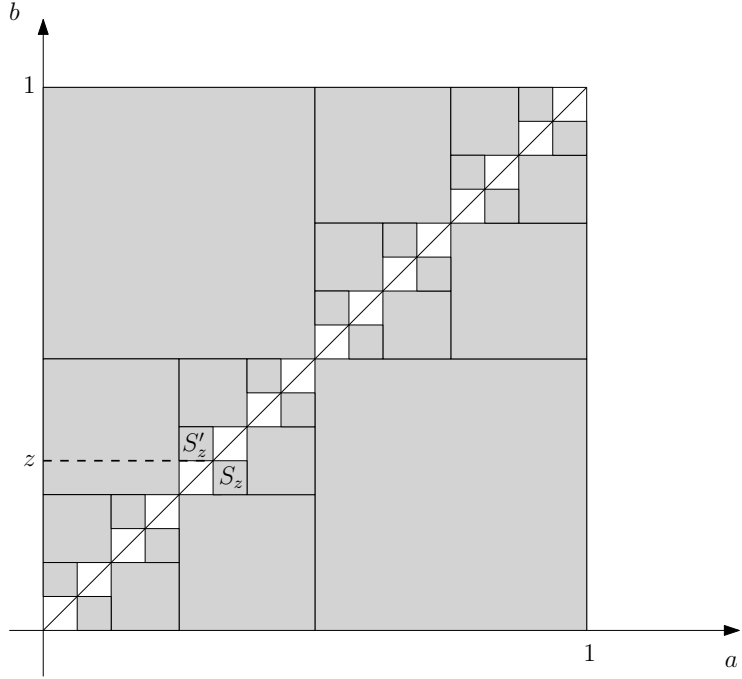}
  \caption{The square $[0,1] \times [0,1]$ (except the diagonal $y=x$) partitioned into smaller squares $S^\prime_z$ and $S_z$.}
  \label{fig:lower_bound_omega}
\end{figure}
We now have the following:
$$\iint_{[0,1] \times [0,1]} f(a, b) \,da \,db = \sum_z \iint_{S_z \cup S^\prime_z} f(a, b) \,da \,db = 1 \enspace .$$
There must be at least one $z$, $0 < z < 1$, such that
$$\iint_{S_z \cup S^\prime_z} f(a, b) \,da \,db > 0 \enspace ,$$
and we also have the following for all $i$:
$$P(z \in I_i) \geq \iint_{S_z \cup S^\prime_z} f(a, b) \,da \,db\enspace .$$
According to the weak law of large numbers, $z$ will be contained in $\Omega(n)$ intervals whp.
%\hfill \qed
\end{proof}

We now present a proof of the main theorem of the paper:
%\begin{theorem}
%\label{thm:maintheorem}
%There is an online algorithm for the $b$-OVERLAP-COLORING problem
%processing the intervals in increasing order of their starting points
%such that
%\begin{equation}
%\frac{\chi^\prime_b}{\chi_b} \leq 1 + O(bn^{-\frac{1}{2}}) \mbox{ whp}
%\end{equation}
%The time for processing one interval is $O(\log n)$.
%\end{theorem}

\begin{proof}[Proof (Theorem~\ref{thm:maintheorem})]
  Algorithm B is an online algorithm using $O(\log n)$ time per item according to Lemma~\ref{lem:online_implementation}.
  From Lemma~\ref{lem:chi_prime_h} and
  \ref{lem:chains}, we conclude that $\chi^\prime_h
  \leq \frac{\omega^\prime}{h} + O(\sqrt{n})$ whp. The minimum number
  of colors $\chi_h$ satisfies $\chi_h \geq \frac{\omega^\prime}{h}$,
  so we now have the following:
\begin{equation}
  \frac{\chi^\prime_h}{\chi_h} \leq 1+\frac{h}{\omega^\prime} O(\sqrt{n})\mbox{ whp} \enspace .
\end{equation}
Finally, we use $\omega^\prime=\Omega(n)$
whp according to Lemma~\ref{lem:lower_bound_omega}.
%\hfill \qed
\end{proof}

A corollary of Theorem~\ref{thm:maintheorem} is that
$\frac{\chi^\prime_h}{\chi_h}$ converges to $1$ in probability if $h =
o(\sqrt{n})$. It is not hard to prove that $\chi^\prime_h \leq \omega^\prime$, which implies $\frac{\chi^\prime_h}{\chi_h} \leq h$, so as another corollary, the expected value of $\frac{\chi^\prime_h}{\chi_h}$ converges to $1$ if $h$ is a constant.

\section{Experiments}\label{sec:experiments}

We have performed some experiments to verify the theoretical results and to examine the underlying constants for the big O notation. For the first type of experiments, we have used the unknown distribution model introduced in Definition~\ref{def:unknown} with a uniform distribution on $\{(a,b) \in [0, 1] \times [0, 1]: |a-b| \leq \ell\}$ for some number $\ell$ as the "unknown" distribution. In other words, we are choosing an interval with length up to $\ell$ uniformly at random. We use the notation $U(\ell)$ for this type of experiment.

For the second type of experiments, we go beyond the unknown distribution model and choose the center and the length of an interval independently using two normal (Gaussian) distributions (if the length is negative, then we ignore it and pick a new length). This means that any real number can be an interval endpoint. The notation $N(\mu_c, \sigma_c^2, \mu_l, \sigma_l^2)$ is used for the second type of experiments, where $\mu_c$ and $\sigma_c$ are the mean and the standard deviation for the center of an interval, and $\mu_l$ and $\sigma_l$ are the corresponding entities for the length of an interval. We go beyond the unknown distribution model to look into an even broader setting.

The eight distributions that we have used are $U(\ell)$, $\ell \in \{0.1, 0.3, 0.5, 0.8\}$, and $N(\mu_c, \sigma_c^2, \mu_l, \sigma_l^2)$, $(\mu_c, \sigma_c, \mu_l, \sigma_l) \in \{(0, 1, 1, 0.2),\allowbreak (0, 1, 1, 0.4),\allowbreak (0, 5, 1, 0.2),\allowbreak (0, 5, 1, 0.4)\}$.

The stack capacity has been fixed to $h=5$ for all the experiments. The experiments examine three perspectives corresponding to the three subsections in this section. For every combination of the eight distributions and three perspectives, we have generated $100$ random instances: one instance for each $n$ in the set $\{2000, 4000, 6000, 8000, \allowbreak 10000, \ldots, 200000\}$. Please note that no instances have been reused.

\subsection{Experiments for the Number of Chains}

Lemma~\ref{lem:chains} is a key lemma specifying an upper bound on $c$, i.e., the minimum number of chains that can be formed for an instance of the stacking problem. The values of $c/\sqrt{n}$ have been plotted against $n$ in Fig.~\ref{fig:c-n-graph} and Fig.~\ref{fig:c-n-graph-gaussian} for the uniform type and the Gaussian type of distributions, respectively. 
\begin{figure}[h]
  \centering 
  \includegraphics[scale=0.8]{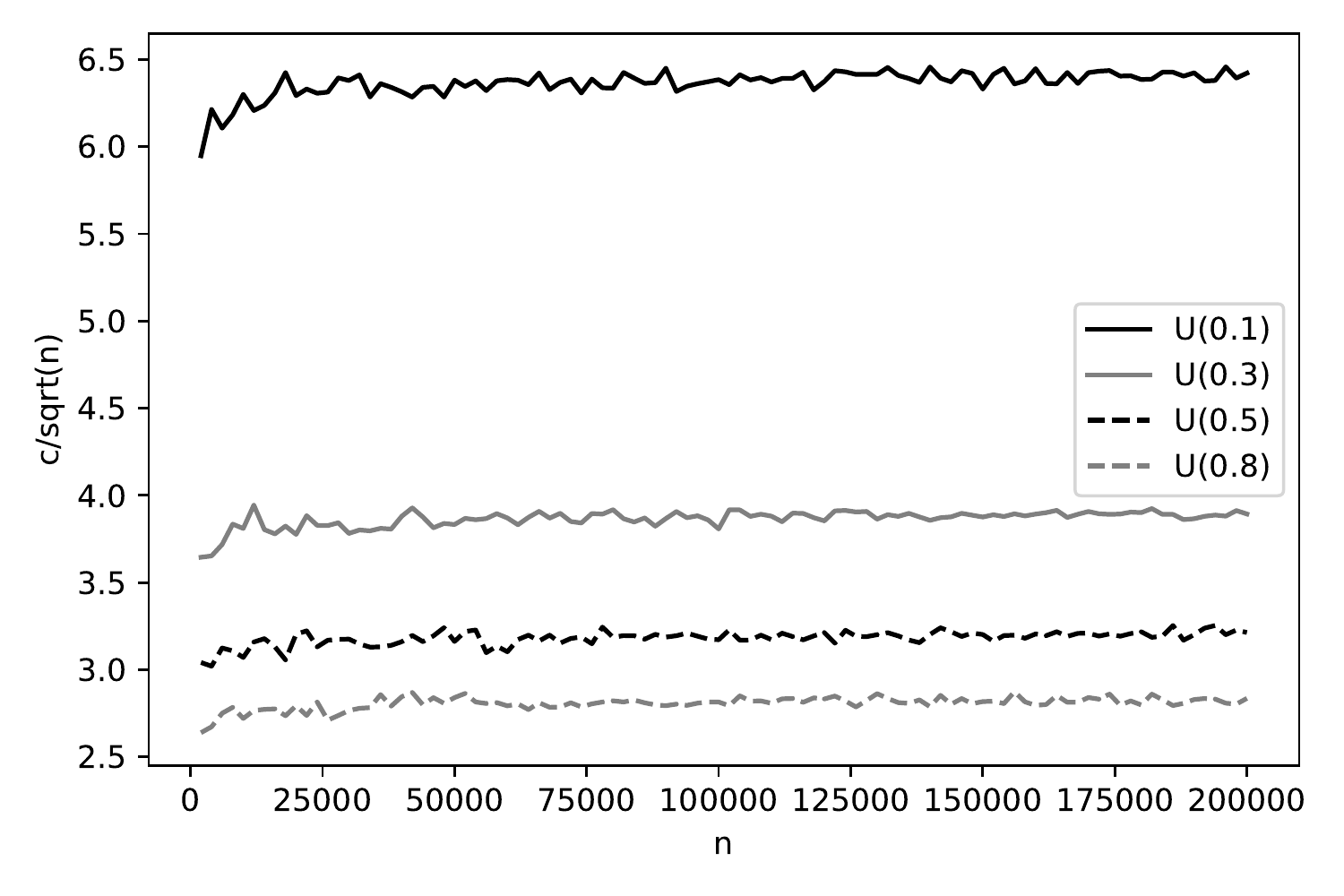}
  \caption{The values of $c/\sqrt{n}$ plotted against $n$ for the uniform type of distributions.}
  \label{fig:c-n-graph}
\end{figure}

\begin{figure}[h]
  \centering 
  \includegraphics[scale=0.8]{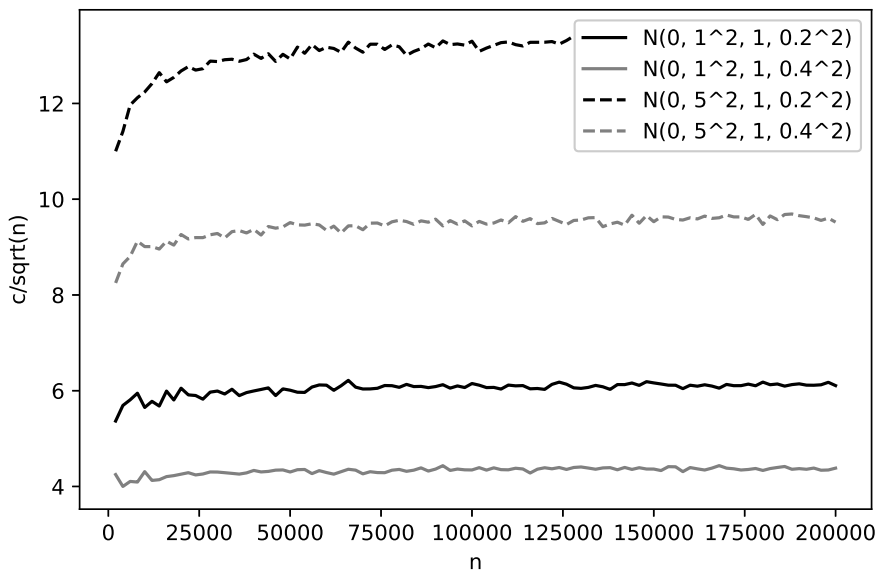}
  \caption{The values of $c/\sqrt{n}$ plotted against $n$ for the Gaussian type of distributions.}
  \label{fig:c-n-graph-gaussian}
\end{figure}
The experiments clearly verify Lemma~\ref{lem:chains} by showing that $c/\sqrt{n}=O(1)$ -- even when we go beyond our unknown distribution model using the Gaussian distributions. The underlying constant $k$ seems to be moderate, and $c/\sqrt{n} \leq k$ holds for {\underline {all}} the instances with $k$ depending on the actual distribution.

\subsection{Convergence Rate Experiments}

We now take a closer experimental look at our main contribution: Theorem~\ref{thm:maintheorem}. Our purpose is to verify the theorem and examine the actual convergence rate for the eight distributions that we consider. Directly inspired by our theorem, we have plotted $\left(\frac{\chi^\prime_h}{(\omega^\prime/h)}-1\right)\sqrt{n}$ against $n$ in Fig.~\ref{fig:relative-error-graph} and Fig.~\ref{fig:relative-error-graph-gaussian}. We remind the reader that $\chi_h \geq \frac{\omega^\prime}{h}$ so $\frac{\chi^\prime_h}{(\omega^\prime/h)}$ is an upper bound on the competitive ratio that we can efficiently compute (as mentioned earlier, we have no efficient procedure for computing $\chi_h$ for $h=5$ at the moment).
\begin{figure}[h]
  \centering 
  \includegraphics[scale=0.8]{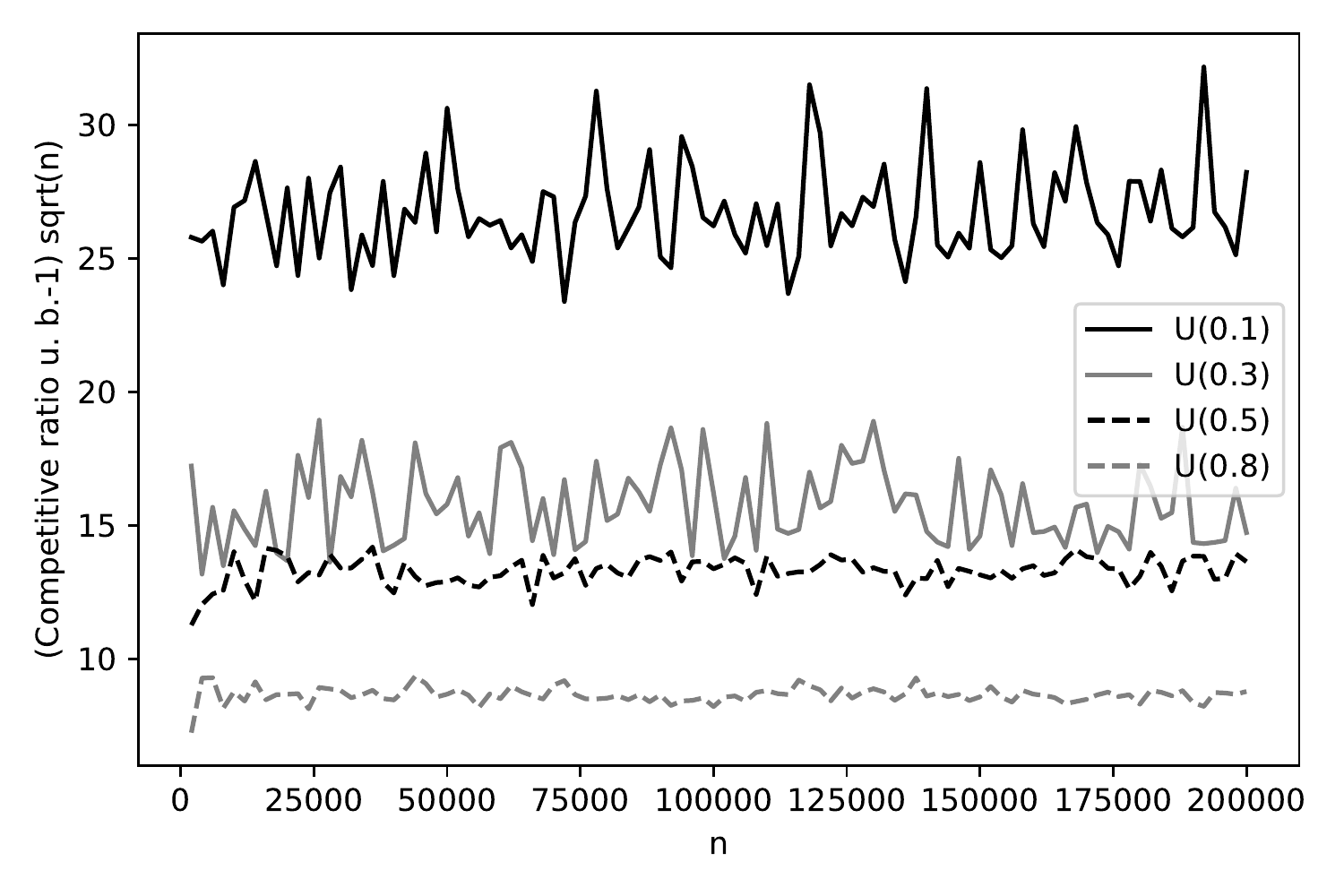}
  \caption{The expression $\left(\frac{\chi^\prime_h}{(\omega^\prime/h)}-1\right)\sqrt{n}$ plotted against $n$ for the uniform type of distributions.}
  \label{fig:relative-error-graph}
\end{figure}

\begin{figure}[h]
  \centering 
  \includegraphics[scale=0.8]{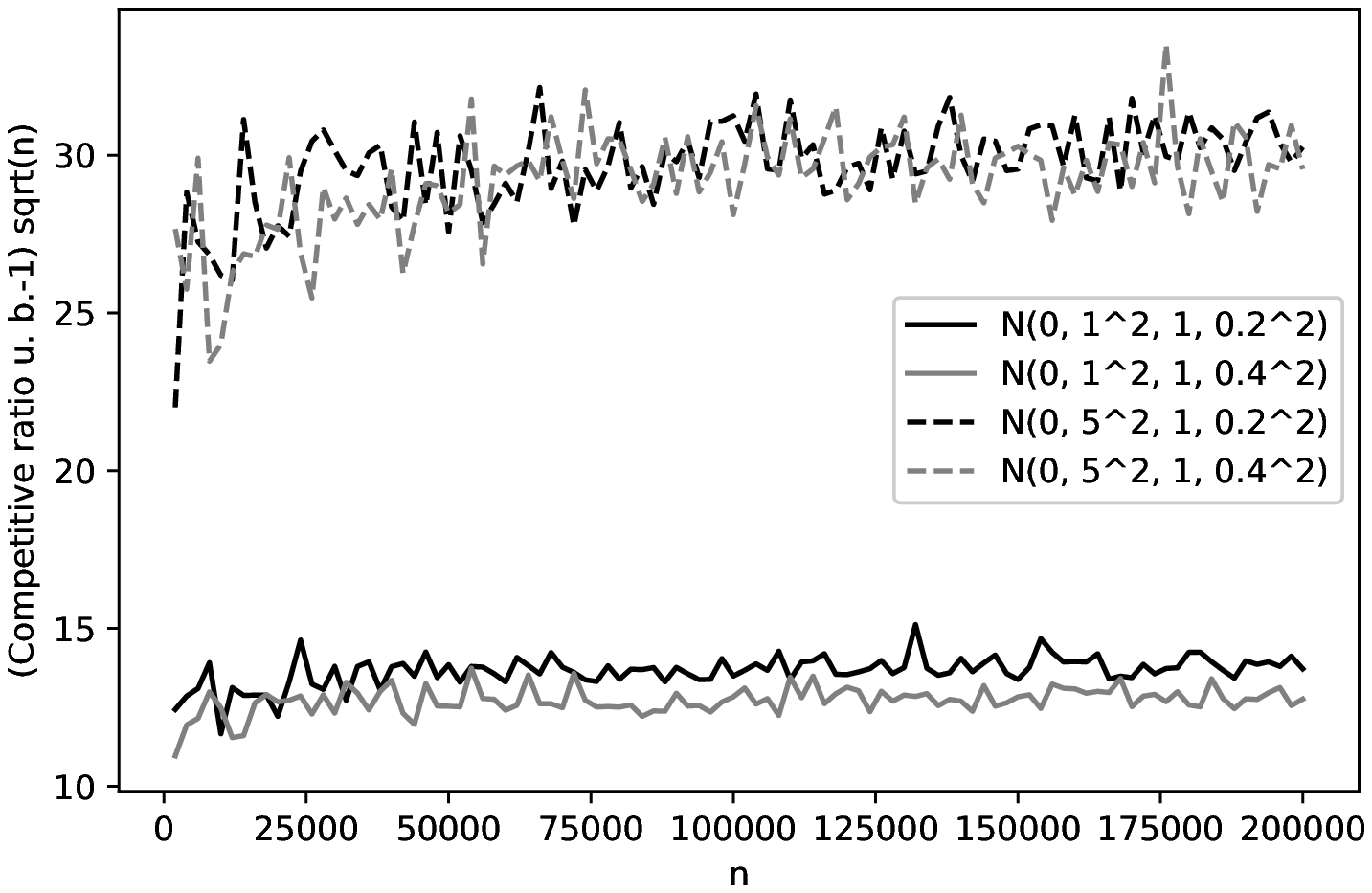}
  \caption{The graph shows the results of the experiments for the expression $\left(\frac{\chi^\prime_h}{(\omega^\prime/h)}-1\right)\sqrt{n}$ for the Gaussian type distributions.}
  \label{fig:relative-error-graph-gaussian}
\end{figure}
Similar to the experiments with the number of chains $c$, we conclude that $\left(\frac{\chi^\prime_h}{(\omega^\prime/h)}-1\right)\sqrt{n}=O(1)$ with an underlying moderate constant $k$. From the graphs, we can se that $\frac{\chi^\prime_h}{(\omega^\prime/h)}\leq 1+k/\sqrt{n}$ is satisfied for all our instances.

\subsection{Competitive Ratio Experiments}

For the sake of completeness, we ran some experiments and plotted the upper bound for the competitive ratio,  $\frac{\chi^\prime_h}{(\omega^\prime/h)}$, against $n$. The results are shown in Fig.~\ref{fig:competitive-ratio-graph} and Fig.~\ref{fig:competitive-ratio-graph-gaussian}.
\begin{figure}[h]
  \centering 
  \includegraphics[scale=0.8]{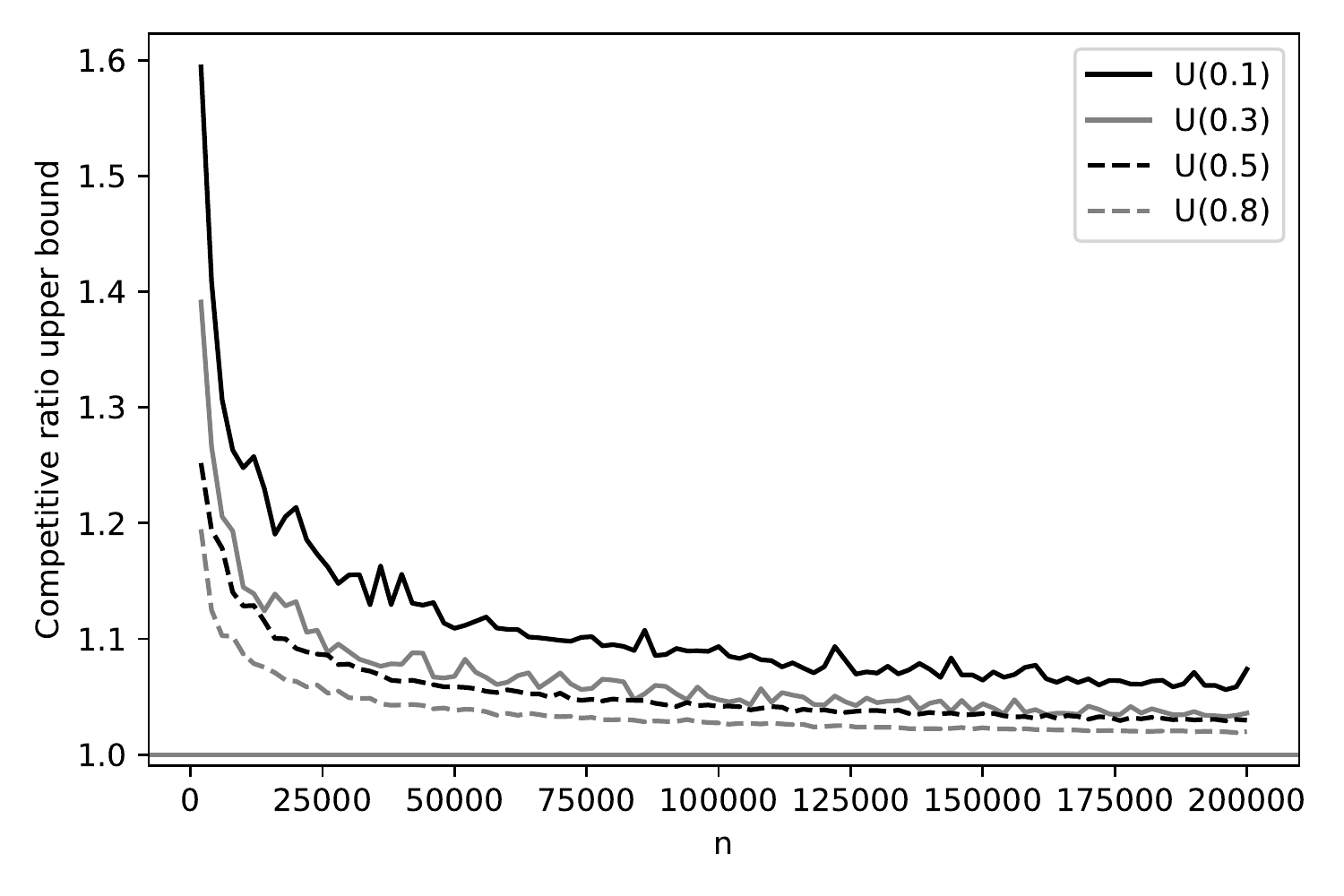}
  \caption{An upper bound for the competitive ratio plotted against $n$ for the uniform type experiments.}
  \label{fig:competitive-ratio-graph}
\end{figure}

\begin{figure}[h]
  \centering 
  \includegraphics[scale=0.8]{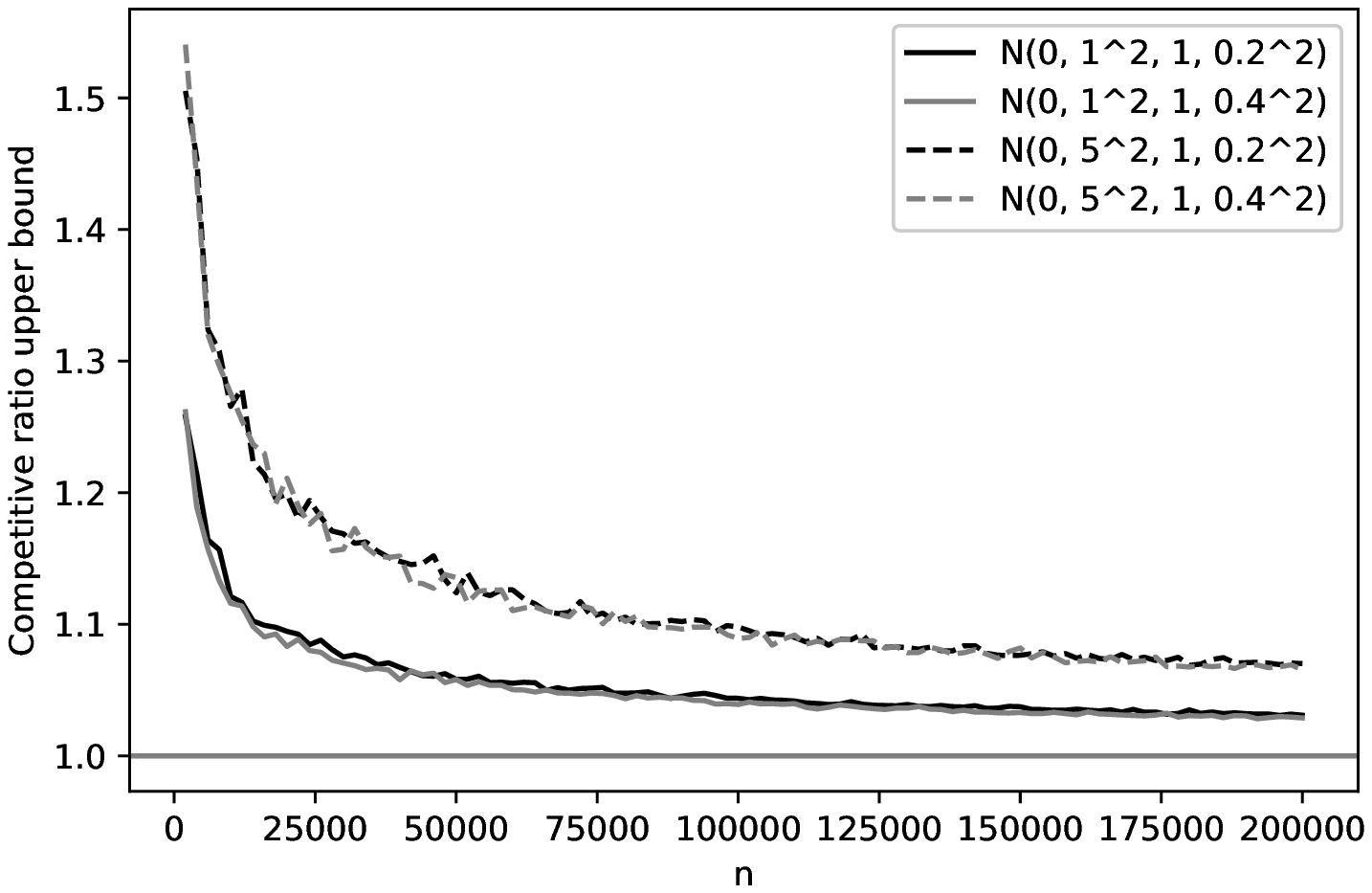}
  \caption{An upper bound for the competitive ratio plotted against $n$ for the Gaussian type experiments.}
  \label{fig:competitive-ratio-graph-gaussian}
\end{figure}
These graphs confirm that the competitive ratio converges to $1$ in probability.

\section*{Conclusion}

We have presented a simple polynomial time online algorithm for stacking with a competitive ratio that converges to $1$ in probability under the unknown distribution model. The main message of our paper is that such an algorithm exists. The experimental part of our paper shows that the results also have practical relevance. We do not think that our algorithm is better than similar algorithms presented in the literature, and we strongly believe that there are other asymptotically optimal algorithms for online stacking. 

\bibliographystyle{plain}

\bibliography{ProbabilisticAnalysisStackingFullPaper}

\end{document}